\documentclass[letterpaper, one column, 12pt]{IEEEtran}

\usepackage{graphics}
\usepackage{amsmath}
\usepackage{amssymb}
\usepackage{graphicx}
\usepackage{cite}
\usepackage{epstopdf}
\newtheorem{theorem}{\bf Theorem}

\newtheorem{corollary}{\bf Corollary}

\DeclareMathAlphabet{\mathssf}{OT1}{cmss}{m}{sl}

\newcommand{\m}[1]{\mathbf{#1}^m}
\newcommand{\lo}[1]{\log_2\left(#1\right)}

\newcommand{\corr}{{\sigma_{SU}}}
\newcommand{\expect}[1]{\mathbb{E}\left[#1\right]}
\newcommand\independent{\protect\mathpalette{\protect\independenT}{\perp}}
\def\independenT#1#2{\mathrel{\rlap{$#1#2$}\mkern2mu{#1#2}}}

\title{Information Embedding meets Distributed Control}
\author{Pulkit $\text{Grover}^\dagger$, Aaron B. $\text{Wagner}^\ddagger$ and Anant $\text{Sahai}^\dagger $\thanks{$\dagger${Wireless Foundations, Department of EECS, University of California at Berkeley. Email: $\{$pulkit, sahai$\}$\;@\;eecs.berkeley.edu}. $\ddagger$ School of Electrical and Computer Engineering, Cornell University. Email: wagner\;@\;ece.cornell.edu. An abridged version of this paper will be presented at the 2010 Information Theory Workshop (ITW), Cairo, Egypt.    }}

\begin{document}\maketitle
\begin{abstract}
We consider the problem of information embedding where the encoder modifies a white Gaussian host signal in a power-constrained manner to encode the message, and the decoder recovers both the embedded message and the \textit{modified} host signal. This extends the recent work of Sumszyk and Steinberg to the continuous-alphabet Gaussian setting. We show that a dirty-paper-coding based strategy achieves the optimal rate for perfect recovery of the modified host and the message. We also provide bounds for the extension wherein the modified host signal is recovered only to within a specified distortion. When specialized to the zero-rate case, our results provide the tightest known lower bounds on the asymptotic costs for the vector version of a famous open problem in distributed control --- the Witsenhausen counterexample. Using this bound, we characterize the asymptotically optimal costs for the vector Witsenhausen problem numerically to within a factor of $1.3$ for all problem parameters, improving on the earlier best known bound of $2$.
\end{abstract}

\section{Introduction}

The problem of interest in this paper (see Fig.~\ref{fig:blockdgm}) derives its motivation from an information-theoretic standpoint, as well as from a distributed-control perspective. Information-theoretically, the problem is an extension of an information embedding problem recently addressed by Sumszyk and Steinberg~\cite{ReversibleStegotext} --- the encoder ensures that the decoder recovers the \textit{modified} host signal $\m{X}$ perfectly, along with the message. Philosophically, the work in~\cite{ReversibleStegotext} is directed towards understanding how a communication problem changes when an additional requirement, that of the encoder being able to produce a copy of the reconstruction at the decoder, is imposed on the system (in source coding context, the issue was explored by Steinberg in~\cite{SteinbergISIT2008}).  The problem is also closely connected to other information theory problems~\cite{CostaDirtyPaper,KimStateAmplification,KotagiriLaneman,MerhavMasking}. We refer the interested reader to~\cite{WitsenhausenJournal}, where these connections are discussed in detail. 

\begin{figure}[htb]
\begin{center}
\includegraphics[scale=0.5]{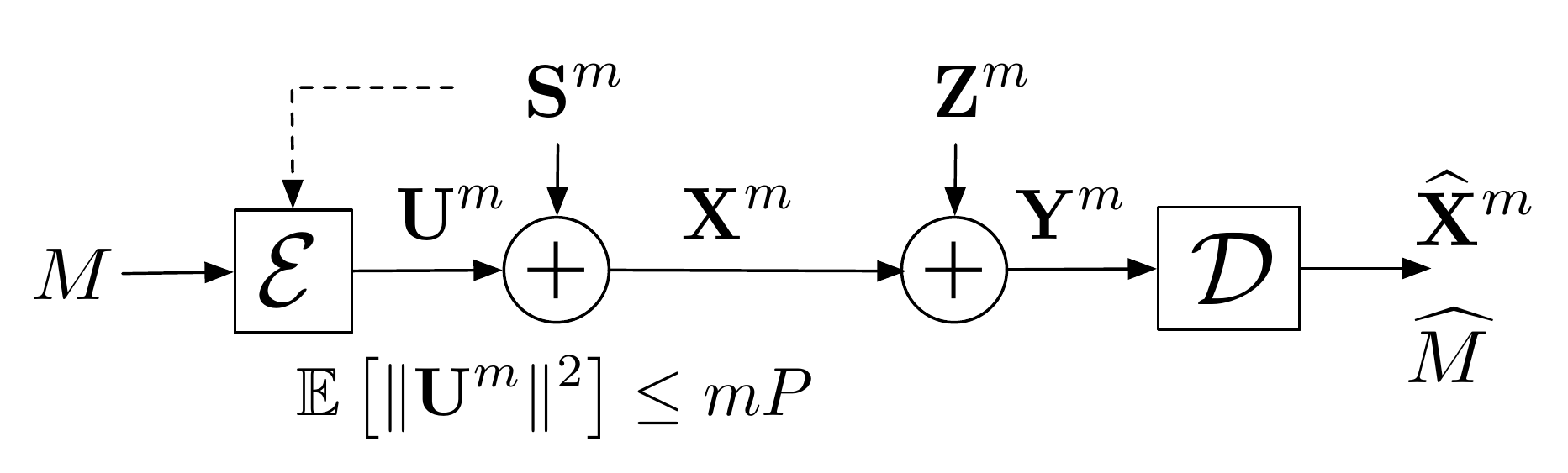}
\caption{The host signal $\m{S}$ is first modified by the encoder using a power constrained input $\m{U}$. The modified host signal $\m{X}$ and the message $M$ are then reconstructed at the decoder. The problem is to find the minimum distortion in reconstruction of $\m{X}$ given $P$, the power constraint, and $R$, the rate of reliable message transmission. }
\label{fig:blockdgm}
\end{center}
\end{figure}

In~\cite{ReversibleStegotext}, the authors assume that the host signal $\m{S}$, the modified host signal (the channel input) $\m{X}$ and the channel output $\m{Y}$ are all \textit{finite-alphabet}. In this paper, we consider the Gaussian version of their problem. The extension is non-trivial~\cite{SteinbergPersonalComm} because simple Fano's inequality-based techniques do not work for the infinite-alphabet formulation. Experience in infinite-alphabet problems might even suggest that  (asymptotic) perfect reconstruction may be impossible because the problem is set in continuous space. Intriguingly, asymptotic perfect reconstruction \textit{is} possible in our problem because the encoder can ensure that the modified host signal takes values in a discrete subset of the continuous space. We provide tight results characterizing the tradeoff between rate and power for perfect reconstruction. As is more natural in a continuous-alphabet setting, we relax the assumption of perfect recovery of the host signal by considering recovery within a specified nonzero distortion, and for this problem we provide upper and lower bounds on the tradeoff between rate, power and average distortion.

The nonzero distortion problem is closely related to the vector version of a famous distributed control problem called the Witsenhausen counterexample~\cite{Witsenhausen68} --- at zero communication rate, the two problems are the same~\cite{WitsenhausenJournal}. The scalar counterexample is believed to be quite challenging (see~\cite{WitsenhausenJournal} for a survey of prior results showing why it is believed to be so). As a conceptual simplification, Grover and Sahai~\cite{WitsenhausenJournal} considered the long-blocklength limit of the counterexample. Further, they relaxed the requirement of obtaining a provably optimal strategy to the weaker objective of obtaining strategies that attain within a constant factor of the optimal cost. For the weighted sum of power and average distortion costs (see Section~\ref{sec:probstat}), they then show that dirty-paper coding techniques attain within a factor of $2$ of the optimal cost for all problem parameters (\textit{i.e.} the weights and the variances of the random variables). Backing off from the infinite blocklength limit, Grover, Park and Sahai~\cite{ConComPaper} then showed that similar constant-factor results can also be obtained for finite vector lengths, including the scalar case. The achievable strategy, which yields the upper bounds, now uses lattices instead of random codebooks. The lower bound is obtained by applying sphere-packing ideas from information theory to the bound of~\cite{WitsenhausenJournal}.

The lower bound in this paper specialized to rate zero provides an improved lower bound to the costs of the vector Witsenhausen counterexample in the long-blocklength limit. Using this improved bound, we show that the ratio of upper and lower bounds is smaller than $1.3$ regardless of the choice of the weights and the problem parameters. This is an improvement over the previously best known maximum ratio of two~\cite{WitsenhausenJournal}.  

Control theory has long wrestled with the Witsenhausen counterexample. Because it is a canonical problem, a comprehensive distributed-control theory would necessarily include a good understanding of the counterexample. Information-theory has had long-standing canonical problems of its own. In a line of investigation started by Gupta and Kumar~\cite{GuptaKumar}, the question of the capacity of a large wireless network is studied. By restricting attention to obtaining just the scaling of the total capacity, the bar for what might constitute a reasonable information-theoretic solution was lowered. More recently, the calculation of channel capacity to within a finite number of bits\footnote{Our constant-factor results on control costs are closely related to results on bounded gap from capacity. A factor of $2$ approximation in power would be a slightly stronger result than a $\frac{1}{2}$-bit approximation in the capacity of a real channel.} for canonical information-theory problems (e.g. the interference channel~\cite{EtkinOneBit}) has led to significant advances in understanding capacity for larger network communication problems~\cite{DeterministicModel,SalmanThesis}. The recent results on Witsenhausen's counterexample thus raise a parallel hope in distributed control.


\section{Problem Statement}
\label{sec:probstat}
The host signal $\m{S}$ is distributed $\mathcal{N}(0,\sigma^2\mathbb{I})$, and the message $M$ is independent of $\m{S}$ and distributed uniformly over $\{1,2,\ldots,2^{mR}\}$. The encoder $\mathcal{E}_m$ maps $(M,\m{S})$ to $\m{X}$ by additively distorting $\m{S}$ using input $\m{U}$ of average power (for each message) at most $P$, i.e. $\expect{\|\m{S}-\m{X}\|^2}\leq mP$. Additive white Gaussian noise $\m{Z}\sim\mathcal{N}(0,\sigma_z^2\mathbb{I})$, where $\sigma_z^2=1$, is added to $\m{X}$ by the channel. The decoder $\mathcal{D}_m$ maps the channel outputs $\m{Y}$ to both an estimate $\m{\widehat{X}}$ of the modified host signal $\m{X}$ and an estimate $\widehat{M}$ of the message. 

Define the error probability $\epsilon_m(\mathcal{E}_m,\mathcal{D}_m)=\Pr(M\neq \widehat{M})$. For the encoder-decoder sequence $\{\mathcal{E}_m,\mathcal{D}_m\}_{m=1}^\infty$, define the minimum asymptotic distortion $MMSE(P,R)$ as follows
\begin{eqnarray*}
MMSE(P,R) = \underset{\{\mathcal{E}_m,\mathcal{D}_m\}_{m=1}^\infty : \epsilon_m(\mathcal{E}_m,\mathcal{D}_m)\rightarrow 0}{\inf} \;\underset{m\rightarrow\infty}{\lim \sup}\; \frac{1}{m}\expect{\|\m{X}-\m{\widehat{X}}\|^2}.
\end{eqnarray*}
We are interested in the tradeoff between the rate $R$, the power $P$, and $MMSE(P,R)$.

The conventional control-theoretic weighted cost formulation~\cite{Witsenhausen68} defines the total cost to be
\begin{equation}
J = \frac{1}{m}k^2 \|\m{U}\|^2 + \frac{1}{m}\|\m{X}-\m{\widehat{X}}\|^2,
\end{equation}
where $k\in\mathbb{R}^+$. The objective is to minimize the average cost, $\expect{J}$ at rate $R$. The average is taken over the realizations of the host signal, the channel noise, and the message. At $R=0$, the problem is the vector Witsenhausen counterexample~\cite{WitsenhausenJournal}. 

\section{Main Results}
\subsection{Lower bounds on $MMSE(P,R)$}
\begin{theorem}
\label{thm:newlower}
For the problem as stated in Section~\ref{sec:probstat}, for communicating reliably at rate $R$ with input power $P$, the asymptotic average mean-square error in recovering $\m{X}$ is lower bounded as follows. For $P\geq 2^{2R}-1$,
\begin{equation}
\label{eq:newLowerWithRateThm}
MMSE(P,R)\geq \inf_{\sigma_{SU}}\sup_{\gamma>0} \frac{1}{\gamma^2}\left(\left( \sqrt{\frac{\sigma^2 2^{2R}}{1+\sigma^2+P+2\sigma_{SU}}}  -   \sqrt{(1-\gamma)^2\sigma^2 + \gamma^2 P -2\gamma (1-\gamma)\corr}     \right)^+\right)^2,
\end{equation}
where $\max\left\{-\sigma\sqrt{P}, \frac{2^{2R}-1-P -\sigma^2}{2}\right\}   \leq\sigma_{SU}\leq\sigma\sqrt{P}$. For $P<2^{2R}-1$, reliable communication at rate $R$ is not possible. 
\end{theorem}

\begin{corollary}\label{coro:wit}
For the vector Witsenhausen problem with $\expect{\|\m{U}\|^2}\leq mP$, the following is a lower bound on the $MMSE$ in the estimation of $\m{X}$.
\begin{equation}
\label{eq:newLower}
MMSE(P,0)\geq \inf_{\sigma_{SU}}\sup_{\gamma>0} \frac{1}{\gamma^2}\left(\left( \sqrt{\frac{\sigma^2}{1+\sigma^2+P+2\sigma_{SU}}}  -   \sqrt{(1-\gamma)^2\sigma^2 + \gamma^2 P -2\gamma (1-\gamma)\corr}     \right)^+\right)^2.
\end{equation}
where $\sigma_{SU}\in [-\sigma\sqrt{P},\sigma\sqrt{P}]$. 
\end{corollary}
\begin{proof}\textit{[Of Theorem~\ref{thm:newlower}]}
For conceptual clarity, we first derive the result for the case $R=0$ (Corollary~\ref{coro:wit}). The tools developed are then used to derive the lower bound for $R> 0$. 
 
\begin{proof}\textit{[Of Corollary~\ref{coro:wit}]}

For any chosen pair of encoding map $\mathcal{E}_m$ and decoding map $\mathcal{D}_m$, there is a Markov chain $\m{S}\rightarrow \m{X}\rightarrow \m{Y}\rightarrow \m{\widehat{X}}$. Using the data-processing inequality
\begin{equation}
\label{eq:dpi}
I(\m{S};\m{\widehat{X}}) \leq I(\m{X};\m{Y}).
\end{equation}
The terms in the inequality can be bounded by single letter expressions as follows. Define $Q$ as a random variable uniformly distributed over $\{1,2,\ldots,m\}$. Define $S=S_Q$, $U=U_Q$, $X=X_Q$, $Z=Z_Q$, $Y=Y_Q$ and $\widehat{X}=\widehat{X}_Q$. Then, 
\begin{eqnarray}
I(\m{X};\m{Y}) &=& h(\m{Y}) - h(\m{Y}|\m{X})\nonumber\\
&\overset{(a)}{\leq}& \sum_ih(Y_i) - h(\m{Y}|\m{X})\nonumber\\
&=& \sum_ih(Y_i) - h(Y_i|X_{i})\nonumber\\
&=& \sum_iI(X_{i};Y_i)\nonumber\\
&=& m I(X;Y|Q)\nonumber\\
&= & m\left( h(Y|Q) - h(Y|X,Q)\right)\nonumber\\
&\leq & m\left( h(Y) - h(Y|X,Q)\right)\nonumber\\
&\overset{(b)}{=} &m\left( h(Y) - h(Y|X)\right) = m I(X;Y),
\label{eq:ixy}
\end{eqnarray}
where $(a)$ follows from an application of the chain-rule for entropy followed by using the fact that conditioning reduces entropy, and $(b)$ follows from the observation that the additive noise $Z_i$ is iid across time, and independent of the input $X_i$ (thus $Y\independent Q|X$). Also, 
\begin{eqnarray}
I(\m{S};\m{\widehat{X}}) &=& h(\m{S}) - h(\m{S}|\m{\widehat{X}}) \nonumber\\
& = &  \sum_i h(S_i) - h(\m{S}|\m{\widehat{X}}) \nonumber\\
& \overset{(a)}{\geq} &  \sum_i \left(h(S_i) - h(S_{i}|\widehat{X}_{i})\right) \nonumber\\
& = &  \sum_i I(S_{i};\widehat{X}_{i}) = m I(S;\widehat{X}|Q)\nonumber\\
& = & m\left(h(S|Q) - h(S|\widehat{X},Q)\right)\nonumber\\
& \overset{(b)}{\geq} & m\left( h(S) - h(S|\widehat{X})\right) = mI(S;\widehat{X}),
\label{eq:isx}
\end{eqnarray}
where $(a)$ and $(b)$ again follow from the fact that conditioning reduces entropy, and $(b)$ also uses the observation that since $S_i$ are iid, $S$, $S_i$, and $S|Q=q$ are distributed identically.

Now, using~\eqref{eq:dpi},~\eqref{eq:ixy} and~\eqref{eq:isx},
\begin{equation}
\label{eq:bigineq}
mI(S;\widehat{X})\leq I(\m{S};\m{\widehat{X}}) \leq I(\m{X};\m{Y}) \leq mI(X;Y).
\end{equation}
Also observe that from the definitions of $S$, $X$, $\widehat{X}$ and $Y$, $\expect{d(\m{S},\m{X})}=\expect{d(S,X)}$, and $\expect{d(\m{X},\m{\widehat{X}})}=\expect{d(X,\widehat{X})}$. 

Using the Cauchy-Schwartz inequality, the correlation $\sigma_{SU}=\expect{SU}$  must satisfy the following constraint,
\begin{equation}
\label{eq:sigmasu}
|\sigma_{SU}|= |\expect{SU}|\leq \sqrt{\expect{S^2}}\sqrt{\expect{U^2}}\leq \sigma\sqrt{P}.
\end{equation}
Also,
\begin{equation}
\label{eq:xpow}
\expect{X^2}=\expect{(S+U)^2}= \sigma^2 + P + 2\sigma_{SU}.
\end{equation} 
Since $Z=Y-X\independent X$, and a Gaussian input distribution maximizes the mutual information across an average-power-constrained AWGN channel,
\begin{equation}
\label{eq:ixycompute}
I(X;Y)\leq \frac{1}{2}\lo{1+\frac{P+\sigma^2+2\sigma_{SU}}{1}}.
\end{equation}
\begin{eqnarray}
I(S;\widehat{X}) &=& h(S) - h(S|\widehat{X})\nonumber\\
& = & h(S) - h(S-\gamma \widehat{X}|\widehat{X})\;\forall \gamma \nonumber\\
& \overset{(a)}{\geq} & h(S) - h(S-\gamma \widehat{X})\nonumber\\
& = & \frac{1}{2}\lo{2\pi e\sigma^2} - h(S-\gamma \widehat{X}),
\label{eq:lossy}
\end{eqnarray}
where $(a)$ follows from the fact that conditioning reduces entropy. Also note here that the result holds for any $\gamma>0$, and in particular, $\gamma$ can depend on $\sigma_{SU}$. Now, 
\begin{eqnarray}
\label{eq:sminusxhat}
 h(S-\gamma\widehat{X}) &= & h(S-\gamma(\widehat{X}-X) -\gamma X)\nonumber\\
& = & h\left(S -\gamma (\widehat{X}-X)- \gamma S -\gamma U \right)\nonumber \\
& = & h\left( (1-\gamma)S-\gamma U -\gamma (\widehat{X}-X) \right).
\end{eqnarray}
The second moment of a sum of two random variables $A$ and $B$ can be bounded as follows
\begin{eqnarray}
\label{eq:aplusbsq}
\expect{(A+B)^2} & = & \expect{A^2}+\expect{B^2}+2\expect{AB}\nonumber\\
& \overset{\text{Cauchy-Schwartz ineq.}}{\leq} & \expect{A^2}+\expect{B^2}+2\sqrt{\expect{A^2}}\sqrt{\expect{B^2}}\nonumber\\
& = & (\sqrt{\expect{A^2}}  + \sqrt{\expect{B^2}})^2,
\end{eqnarray}
with equality when $A$ and $B$ are aligned, i.e. $A=\lambda B$ for some $\lambda\in \mathbb{R}$. For the random variable under consideration in~\eqref{eq:sminusxhat}, choosing $A = (1-\gamma)S-\gamma U$, and $B = -\gamma (\widehat{X}-X)$ in~\eqref{eq:aplusbsq}
\begin{eqnarray}
\label{eq:align}
\expect{\left((1-\gamma)S - \gamma U -\gamma (\widehat{X}-X) \right)^2 } \leq \left(\sqrt{(1-\gamma)^2\sigma^2 + \gamma^2 P -2\gamma (1-\gamma)\corr} + \gamma \sqrt{\expect{(\widehat{X}-X)^2}}\right)^2.
\end{eqnarray}
Equality is obtained by aligning\footnote{In general, since $\m{\widehat{X}}$ is a function of $\m{Y}$, this alignment is not actually possible when the recovery of $\m{X}$ is not exact. The derived bound is therefore loose.} $X-\widehat{X}$ with $(1-\gamma)S - \gamma U$. Thus,
\begin{eqnarray}
I(S;\widehat{X}) &\geq &\frac{1}{2}\lo{2\pi e\sigma^2} - h(S-\gamma\widehat{X})\nonumber\\
&\geq & \frac{1}{2}\lo{\frac{\sigma^2}{\left(\sqrt{(1-\gamma)^2\sigma^2 + \gamma^2 P -2\gamma (1-\gamma)\corr} + \gamma \sqrt{\expect{(\widehat{X}-X)^2}}\right)^2}}.
\label{eq:isxhatcompute}
\end{eqnarray}
Using~\eqref{eq:bigineq}, $I(S;\widehat{X})\leq I(X;Y)$. Using the lower bound on $I(S;\widehat{X})$ from~\eqref{eq:isxhatcompute} and the upper bound on $I(X;Y)$ from~\eqref{eq:ixycompute}, we get
\begin{eqnarray*}
\frac{1}{2}\lo{\frac{\sigma^2}{\left(\sqrt{(1-\gamma)^2\sigma^2 + \gamma^2 P -2\gamma (1-\gamma)\corr} + \gamma \sqrt{\expect{(\widehat{X}-X)^2}}\right)^2}} \leq \frac{1}{2}\lo{1+\frac{P+\sigma^2+2\sigma_{SU}}{1}},
\end{eqnarray*}
for the choice of $\mathcal{E}_m$ and $\mathcal{D}_m$. Since $\lo{\cdot{}}$ is a monotonically increasing function, 
\begin{eqnarray*}
\frac{\sigma^2}{\left(\sqrt{(1-\gamma)^2\sigma^2 + \gamma^2 P -2\gamma (1-\gamma)\corr} + \gamma \sqrt{\expect{(\widehat{X}-X)^2}}\right)^2}\leq 1+P+\sigma^2+2\sigma_{SU}\\
\text{i.e.}\;\;\left(\sqrt{(1-\gamma)^2\sigma^2 + \gamma^2 P -2\gamma (1-\gamma)\corr} + \gamma \sqrt{\expect{(\widehat{X}-X)^2}}\right)^2\geq \frac{\sigma^2}{1+P+\sigma^2+2\sigma_{SU}},\\
\text{Since $\gamma>0$,}\;\;\gamma \sqrt{\expect{(\widehat{X}-X)^2}}\geq \sqrt{\frac{\sigma^2}{1+P+\sigma^2+2\sigma_{SU}}}- \sqrt{(1-\gamma)^2\sigma^2 + \gamma^2 P -2\gamma (1-\gamma)\corr}. 
\end{eqnarray*}
Because the RHS may not be positive, we take the maximum of zero and the RHS and obtain the following lower bound for $\mathcal{E}_m$ and $\mathcal{D}_m$. 
\begin{equation}
\expect{(\widehat{X}-X)^2}\geq  \frac{1}{\gamma^2}\left(\left(\sqrt{\frac{\sigma^2}{1+P+\sigma^2+2\sigma_{SU}}}- \sqrt{(1-\gamma)^2\sigma^2 + \gamma^2 P -2\gamma (1-\gamma)\corr}\right)^+\right)^2.
\end{equation}
Because the bound holds for every $\gamma>0$,
\begin{equation}
\expect{(\widehat{X}-X)^2}\geq \sup_{\gamma>0} \frac{1}{\gamma^2}\left(\left(\sqrt{\frac{\sigma^2}{1+P+\sigma^2+2\sigma_{SU}}}- \sqrt{(1-\gamma)^2\sigma^2 + \gamma^2 P -2\gamma (1-\gamma)\corr}\right)^+\right)^2,
\end{equation}
for the chosen $\mathcal{E}_m$ and $\mathcal{D}_m$. Now, from~\eqref{eq:sigmasu}, $\sigma_{SU}$ can take values in $[-\sigma\sqrt{P},\sigma\sqrt{P}]$. Because the lower bound depends on $\mathcal{E}_m$ and $\mathcal{D}_m$ only through $\sigma_{SU}$, we obtain the following lower bound for all $\mathcal{E}_m$ and $\mathcal{D}_m$,
\begin{equation}
\expect{(\widehat{X}-X)^2}\geq \inf_{|\sigma_{SU}|\leq \sigma\sqrt{P}}\sup_{\gamma>0} \frac{1}{\gamma^2}\left(\left(\sqrt{\frac{\sigma^2}{1+P+\sigma^2+2\sigma_{SU}}}- \sqrt{(1-\gamma)^2\sigma^2 + \gamma^2 P -2\gamma (1-\gamma)\corr}\right)^+\right)^2,
\end{equation}
which proves Corollary~\ref{coro:wit}. Notice that we did not take limits in $m$ anywhere, and hence the lower bound holds for all values of $m$.
\end{proof}

\subsection*{The case of nonzero rate}
To prove Theorem~\ref{thm:newlower}, consider now the problem when the encoder wants to also communicate a message $M$ reliably to the decoder at rate $R$.

Using Fano's inequality, since $\Pr(M\neq \widehat{M})=\epsilon_m\rightarrow 0$ as $m\rightarrow\infty$, $H(M|\widehat{M})\leq m\delta_m$ where $\delta_m\rightarrow 0$. Thus,
\begin{eqnarray}
\nonumber I(M;\widehat{M})&=& H(M)-H(M|\widehat{M})\\
\nonumber &=& mR - H(M|\widehat{M})\\
&\geq & mR - m\delta_m = m(R-\delta_m).
\label{eq:fano}
\end{eqnarray}
As before, we consider a mutual information inequality that follows directly from the Markov chain $(M,\m{S})\rightarrow \m{X}\rightarrow \m{Y}\rightarrow (\m{\widehat{X}},\widehat{M})$ :
\begin{equation}
\label{eq:dpi2}
I(M,\m{S};\widehat{M},\m{\widehat{X}}) \leq I(\m{X};\m{Y}).
\end{equation}
The RHS can be bounded above as in~\eqref{eq:ixy}. For the LHS,
\begin{eqnarray}
\nonumber 
I(M,\m{S};\widehat{M},\m{\widehat{X}}) &=& I(M;\widehat{M},\m{\widehat{X}}) + I(\m{S};\widehat{M},\m{\widehat{X}}|M)\\
\nonumber & \geq &I(M;\widehat{M}) +   I(\m{S};\widehat{M},\m{\widehat{X}}|M)\\
\nonumber & = &I(M;\widehat{M}) +  h(\m{S}|M) -h(\m{S}|\widehat{M},\m{\widehat{X}},M)\\
\nonumber & \overset{\m{S}\independent M}{=} &I(M;\widehat{M}) +  h(\m{S}) -h(\m{S}|\widehat{M},\m{\widehat{X}},M)\\
\nonumber  &\geq &I(M;\widehat{M}) +  h(\m{S}) -h(\m{S}|\m{\widehat{X}})\\
&\geq &I(M;\widehat{M}) +  I(\m{S};\m{\widehat{X}})\nonumber\\
&\overset{\text{using~\eqref{eq:isx}}}{\geq} &I(M;\widehat{M}) +  mI(S;\widehat{X}).
\label{eq:breaking}
\end{eqnarray}
From~\eqref{eq:fano},~\eqref{eq:dpi2} and~\eqref{eq:breaking}, we obtain
\begin{eqnarray}
\label{eq:bigineq2}
\nonumber m(R-\delta_m)+  mI(S;\widehat{X})& \overset{\text{using}~\eqref{eq:fano}}{\leq} &I(M;\widehat{M}) +  mI(S;\widehat{X})\\
\nonumber & \overset{\text{using}~\eqref{eq:breaking}}{\leq} &I(M,\m{S};\widehat{M},\m{\widehat{X}}) \\
 &\overset{\text{using}~\eqref{eq:dpi2}}{\leq} & I(\m{X};\m{Y})\overset{\text{using}~\eqref{eq:ixy}}{\leq} m I(X;Y).
\end{eqnarray}
$I(X;Y)$ and $I(S;\widehat{X})$ can be bounded as before in~\eqref{eq:ixycompute} and~\eqref{eq:isxhatcompute}. Observing that as $m\rightarrow\infty$, $\delta_m\rightarrow 0$, we get the following lower bound on the $MMSE$ for nonzero rate,
\begin{equation}
\label{eq:newLowerWithRate}
MMSE(P,R)\geq \inf_{\sigma_{SU}}\sup_{\gamma>0} \frac{1}{\gamma^2}\left(\left( \sqrt{\frac{\sigma^2 2^{2R}}{1+\sigma^2+P+2\sigma_{SU}}}  -   \sqrt{(1-\gamma)^2\sigma^2 + \gamma^2 P -2\gamma (1-\gamma)\corr}     \right)^+\right)^2.
\end{equation}
In the limit $\delta_m\rightarrow 0$, we require from~\eqref{eq:bigineq2} that $I(X;Y)\geq R$. This gives the following constraint on $\sigma_{SU}$,
\begin{eqnarray}
\frac{1}{2}\lo{1+P+\sigma^2+2\sigma_{SU}}\geq R \nonumber \\
\text{i.e.}\; \sigma_{SU}\geq \frac{2^{2R}-1-P-\sigma^2}{2},
\end{eqnarray}
yielding (in conjunction with~\eqref{eq:sigmasu}) the constraint on $\sigma_{SU}$ in Theorem~\ref{thm:newlower}. The constraint on $P$ in the Theorem follows from Costa's result~\cite{CostaDirtyPaper}, because the rate $R$ must be smaller than the capacity over a power constrained AWGN channel with known interference, $\frac{1}{2}\lo{1+P}$. 
\end{proof}
It is insightful to see how the lower bound in Corollary~\ref{coro:wit} is an improvement over that in~\cite{WitsenhausenJournal}. The lower bound in~\cite{WitsenhausenJournal} is given by
\begin{equation}
\label{eq:oldbounddd}
MMSE(P,0)\geq \left(\left(  \sqrt{\frac{\sigma^2}{\sigma^2+P+2\sigma\sqrt{P}+1}}   -\sqrt{P}\right)^+\right)^2,
\end{equation}
which again holds for all $m$. Because any $\gamma$ provides a valid lower bound in Corollary~\ref{coro:wit}, choosing $\gamma=1$ in Corollary~\ref{coro:wit} provides the following (loosened) bound,
\begin{equation}
MMSE(P,0)\geq \inf_{|\sigma_{SU}|\leq \sigma\sqrt{P}}\left(\left(  \sqrt{\frac{\sigma^2}{\sigma^2+P+2\sigma_{SU}+1}}   -\sqrt{P}\right)^+\right)^2,
\end{equation}
which is minimized for $\sigma_{SU}=\sigma\sqrt{P}$. This immediately yields the lower bound~\eqref{eq:oldbounddd} of~\cite{WitsenhausenJournal}.


\subsection{The upper bound and the tightness at $MMSE=0$}
We use the combination of linear and dirty-paper coding strategies of~\cite{WitsenhausenJournal}, except that we communicate a message at rate $R$ as well. We summarize the strategy briefly, and refer the interested reader to~\cite{WitsenhausenJournal} for a detailed description and analysis of the achievability. 

The encoder divides its input into two parts $\m{U}_{lin}$ and $\m{U}_{dpc}$ of powers $P_{lin}$ and $P_{dpc}$ respectively, such that $P=P_{lin}+P_{dpc}$ (by construction, $\m{U}_{lin}$ and $\m{U}_{dpc}$ turn out to be orthogonal in the limit). We refer to $P_{lin}$ as the \textit{linear} part of the power, and $P_{dpc}$ the \textit{dirty-paper coding} part of the power. The linear part is used to scale the host signal down by a factor $\beta$ (using $\m{U}_{lin}=-\beta\m{S}$) so that the scaled down host signal has variance $\widetilde{\sigma}^2=\sigma^2(1-\beta)^2$, where $\beta^2\sigma^2=P_{lin}$. Using the remaining $P_{dpc}$ power, the transmitter dirty-paper codes against the scaled-down host signal $(1-\beta)\m{S}$ with the DPC parameter $\alpha$~\cite{CostaDirtyPaper} allowed to be arbitrary (unlike in~\cite{CostaDirtyPaper}, where it is eventually chosen to be the MMSE parameter). 

A plain DPC strategy achieves the following rate~\cite[Eq. (6)]{CostaDirtyPaper}
\begin{equation}
R = \frac{1}{2}\lo{\frac{P(P+\sigma^2+1)}{P\sigma^2(1-\alpha)^2+ P+\alpha^2\sigma^2}},
\end{equation}
The strategy recovers $\m{U}+\alpha\m{S}$ at the decoder with high probability. Because we also have a linear part here, the achieved rate is
\begin{equation}
\label{eq:alphabeta}
R = \frac{1}{2}\lo{\frac{P_{dpc}(P_{dpc}+\widetilde{\sigma}^2+1)}{P_{dpc}\widetilde{\sigma}^2(1-\alpha)^2+ P_{dpc}+\alpha^2\widetilde{\sigma}^2}}.
\end{equation}
The decoder now decodes the codeword $\m{U}_{dpc}+\alpha(1-\beta)\m{S} $. It then performs an MMSE estimation for estimating $\m{X}=\m{S}+\m{U}=(1-\beta)\m{S}+\m{U}_{dpc}$ using the channel output $\m{Y}=(1-\beta)\m{S}+\m{U}_{dpc}+\m{Z}$ and the decoded codeword $\alpha(1-\beta)\m{S} + \m{U}_{dpc}$. The obtained $MMSE$ can now be minimized over the choice of $\alpha$ and $\beta$ under the constraint~\eqref{eq:alphabeta}.

\begin{corollary}
\label{coro:match}
For a given power $P$, a combination of linear and DPC-based strategies achieves the maximum rate $C(P)$ in the perfect recovery limit $MMSE(P,R)=0$, where $C(P)$ is given by
\begin{equation}
\label{eq:upRate}
C(P)=\sup_{\sigma_{SU}\in [-\sigma\sqrt{P},0]}\frac{1}{2}\lo{ \frac{   (P\sigma^2 - \sigma_{SU}^2) (1+\sigma^2+P+2\sigma_{SU}) }{\sigma^2(\sigma^2+P+2\sigma_{SU})} }.
\end{equation}

\end{corollary}
\begin{proof}

\textit{The achievability}

The combination of linear and DPC-based strategies of~\cite{WitsenhausenJournal} recovers $\m{U}_{dpc}+\alpha(1-\beta)\m{S}$ at the decoder with high probability. In order to perfectly recover $\m{X}=(1-\beta)\m{S}+\m{U}_{dpc}$, we can use $\alpha=1$, and hence the strategy would achieve a rate of
\begin{equation}
R_{ach} = \sup_{P_{lin},P_{dpc}: P=P_{lin}+P_{dpc}}\frac{1}{2}\lo{\frac{P_{dpc}(P_{dpc}+\widetilde{\sigma}^2+1)}{P_{dpc}+\widetilde{\sigma}^2}},
\end{equation}
where we take a supremum over $P_{lin},P_{dpc}$ such that they sum up to $P$. Let $\sigma_{SU}= -\sigma\sqrt{P_{lin}}$ (note that as $P_{lin}$ varies from $0$ to $P$, $\sigma_{SU}$ varies from $0$ to $-\sigma\sqrt{P}$). Then, $P_{dpc}= P - \frac{\sigma_{SU}^2}{\sigma^2}$, and $P_{dpc}+\widetilde{\sigma}^2 = P_{dpc}+\sigma^2 + P_{lin}- 2\sigma\sqrt{P_{lin}} = P+\sigma^2 +2\sigma_{SU}$. Thus,
\begin{eqnarray}
\label{eq:lowRate}
R_{ach} = \sup_{\sigma_{SU}\in [-\sigma\sqrt{P},0]}\frac{1}{2}\lo{\frac{\left(P - \frac{\sigma_{SU}^2}{\sigma^2}\right)( P+\sigma^2 +2\sigma_{SU}+1)}{ P+\sigma^2 +2\sigma_{SU}}}.
\end{eqnarray}
Simple algebra shows that this expression matches that in Corollary~\ref{coro:match}.
\vspace{0.2in}

\textit{The converse}



Since we are free to choose $\gamma$ in Theorem~\ref{thm:newlower}, let $\gamma = \gamma^* = \frac{\sigma^2+ \sigma_{SU}}{\sigma^2+P+2\sigma_{SU}}$. Then, $1-\gamma^* = \frac{P+ \sigma_{SU}}{\sigma^2+P+2\sigma_{SU}}$. Thus, we get
\begin{equation}
0\geq \inf_{\sigma_{SU}} \frac{1}{\gamma^{*^2}}\left(\left( \sqrt{\frac{\sigma^2 2^{2R}}{1+\sigma^2+P+2\sigma_{SU}}}  -   \sqrt{(1-\gamma^*)^2\sigma^2 + \gamma^{*^2} P -2\gamma^* (1-\gamma^*)\corr}     \right)^+\right)^2.
\end{equation}
It has to be the case that the term inside $(\cdot{})^+$ is non-positive for some value of $\sigma_{SU}$. This immediately yields
\begin{eqnarray*}
2^{2R}&\leq &\sup_{\sigma_{SU}} \frac{1}{\sigma^2}\left((1-\gamma^*)^2\sigma^2 + \gamma^{*^2} P -2\gamma^* (1-\gamma^*)\corr \right) (1+\sigma^2+P+2\sigma_{SU})\\
& = & \sup_{\sigma_{SU}}\frac{1}{\sigma^2}\frac{\left( (P+\sigma_{SU})^2\sigma^2 + (\sigma^2 + \sigma_{SU})^2 P - 2 (P+\sigma_{SU})(\sigma^2+\sigma_{SU})\sigma_{SU}  \right)}{(\sigma^2+P+2\sigma_{SU})^2}(1+\sigma^2+P+2\sigma_{SU})\\
& =&\sup_{\sigma_{SU}} \frac{1}{\sigma^2}\frac{\left(   P^2\sigma^2 - \sigma_{SU}^2\sigma^2 +2P\sigma_{SU}\sigma^2 + P\sigma^4 - P\sigma_{SU}^2 -2\sigma_{SU}^3  \right)}{(\sigma^2+P+2\sigma_{SU})^2} (1+\sigma^2+P+2\sigma_{SU})\\
& =&\sup_{\sigma_{SU}} \frac{1}{\sigma^2}\frac{\left(   (P\sigma^2 - \sigma_{SU}^2) (P+\sigma^2 + 2\sigma_{SU})  \right)}{(\sigma^2+P+2\sigma_{SU})^2} (1+\sigma^2+P+2\sigma_{SU})\\
& =&\sup_{\sigma_{SU}} \frac{   (P\sigma^2 - \sigma_{SU}^2) (1+\sigma^2+P+2\sigma_{SU}) }{\sigma^2(\sigma^2+P+2\sigma_{SU})} 
\end{eqnarray*}
Thus, we get the following upper bound on $C(P)$,
\begin{equation}
\label{eq:upRate2}
C(P)\leq \sup_{\sigma_{SU}\in [-\sigma\sqrt{P},\sigma\sqrt{P}]}\frac{1}{2}\lo{ \frac{   (P\sigma^2 - \sigma_{SU}^2) (1+\sigma^2+P+2\sigma_{SU}) }{\sigma^2(\sigma^2+P+2\sigma_{SU})} }.
\end{equation}
The term $(P\sigma^2-\sigma_{SU}^2)$ is oblivious to the sign of $\sigma_{SU}$. However, the  term 
\begin{equation}
\frac{1+\sigma^2+P+2\sigma_{SU}}{\sigma^2+P+2\sigma_{SU}} =1 +  \frac{1}{\sigma^2+P+2\sigma_{SU}}
\end{equation}
is clearly larger for $\sigma_{SU}<0$ if we fix $|\sigma_{SU}|$. Thus the supremum in~\eqref{eq:upRate2} is attained at some $\sigma_{SU}<0$, and we get
\begin{equation}
\label{eq:upRate3}
C(P)\leq \sup_{\sigma_{SU}\in [-\sigma\sqrt{P},0]}\frac{1}{2}\lo{ \frac{   (P\sigma^2 - \sigma_{SU}^2) (1+\sigma^2+P+2\sigma_{SU}) }{\sigma^2(\sigma^2+P+2\sigma_{SU})} },
\end{equation}
which matches the expression in~\eqref{eq:lowRate}. Thus for perfect reconstruction ($MMSE=0$), the  combination of linear and DPC strategies proposed in~\cite{WitsenhausenJournal} is optimal.
\end{proof}

\section{Numerical results}

Witsenhausen's original control theoretic formulation seeks to minimize the sum
of weighted costs $k^2P+MMSE$. Fig.~\ref{fig:ratio}(b) shows that
asymptotically, the ratio of upper and new lower bounds (from Corollary~\ref{coro:wit}) on the weighted cost is bounded by $1.3$, an improvement over the ratio of $2$
in~\cite{WitsenhausenJournal}. The ridge of ratio $2$
along $\sigma^2=\frac{\sqrt{5}-1}{2}$ present in
Fig.~\ref{fig:ratio}(a) (obtained using the old bound from~\cite{WitsenhausenJournal}) does not exist with the new lower
bound since this small-$k$ regime corresponds to target $MMSE$s close
to zero -- where the new lower bound is tight. This is illustrated in
Fig.~\ref{fig:sigma} (top). Also shown in Fig.~\ref{fig:sigma} (bottom) is the lack of tightness in the bounds at small $P$. The figure explains how this looseness results in the ridge along $k\approx 1.67$ still surviving in the new ratio plot.

Fig.~\ref{fig:MMSEratio} shows the ratio of upper and lower bounds on $MMSE(P,0)$ versus $P$ and $\sigma$. While the ratio with the bound of~\cite{WitsenhausenJournal} was unbounded (Fig.~\ref{fig:MMSEratio}, top), the new ratio is bounded by a factor of $1.5$ (Fig.~\ref{fig:MMSEratio}, bottom). This is again a reflection of the tightness of the bound at small $MMSE$. A flipped perspective is shown in Fig.~\ref{fig:PowerRatio}, where we compute the ratio of upper and lower bounds on required power to attain a specified $MMSE$. As further evidence of the lack of tightness in the small-$P$ (``high distortion'') regime, the ratio of upper and lower bounds on required power diverges to infinity along the path $MMSE=\frac{\sigma^2}{\sigma^2+1}$.

Fig.~\ref{fig:nonzerorate} shows the upper and the lower bounds for $R=0.5$. Again, the bounds are not tight in the small-$P$ regime --- now the looseness is at the lowest power $P=1$ at which communication at $R=0.5$ is possible. As shown in Corollary~\ref{coro:match}, the bounds are still tight at $MMSE=0$. Fig.~\ref{fig:RvsMMSE} shows the upper and lower bounds on $MMSE$ as a function of the rate $R$ for fixed power $P=1$ and $\sigma^2$ equal to the Golden ratio. The figure demonstrates that beyond the maximum rate with zero distortion, the price of increasing rate is an increased distortion in the estimation of $\m{X}$. 

The MATLAB code for these figures can be found in~\cite{CodeForInformationEmbedding}.

\begin{figure}[htb]
\begin{center}
\includegraphics[scale=0.47]{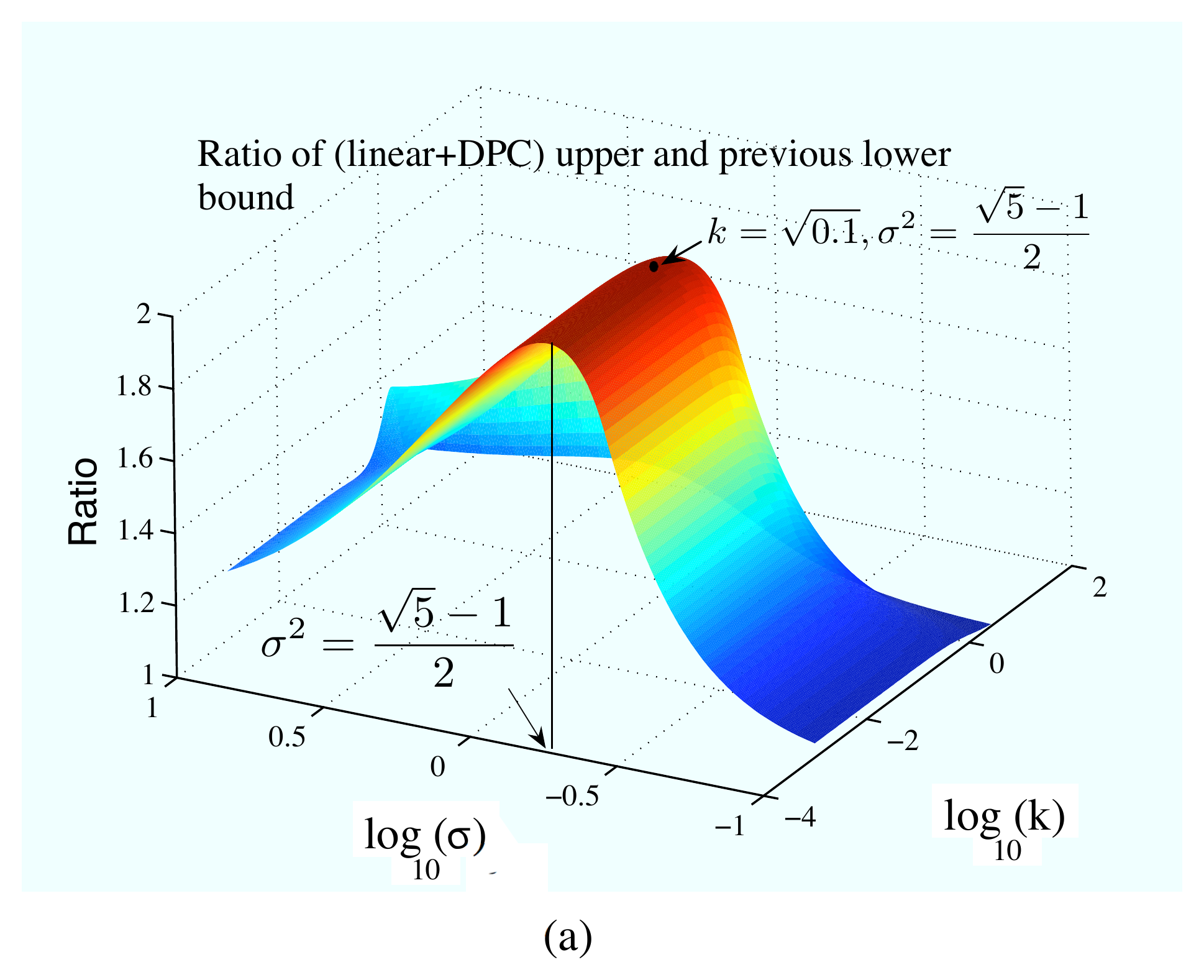}\\\includegraphics[scale=0.47]{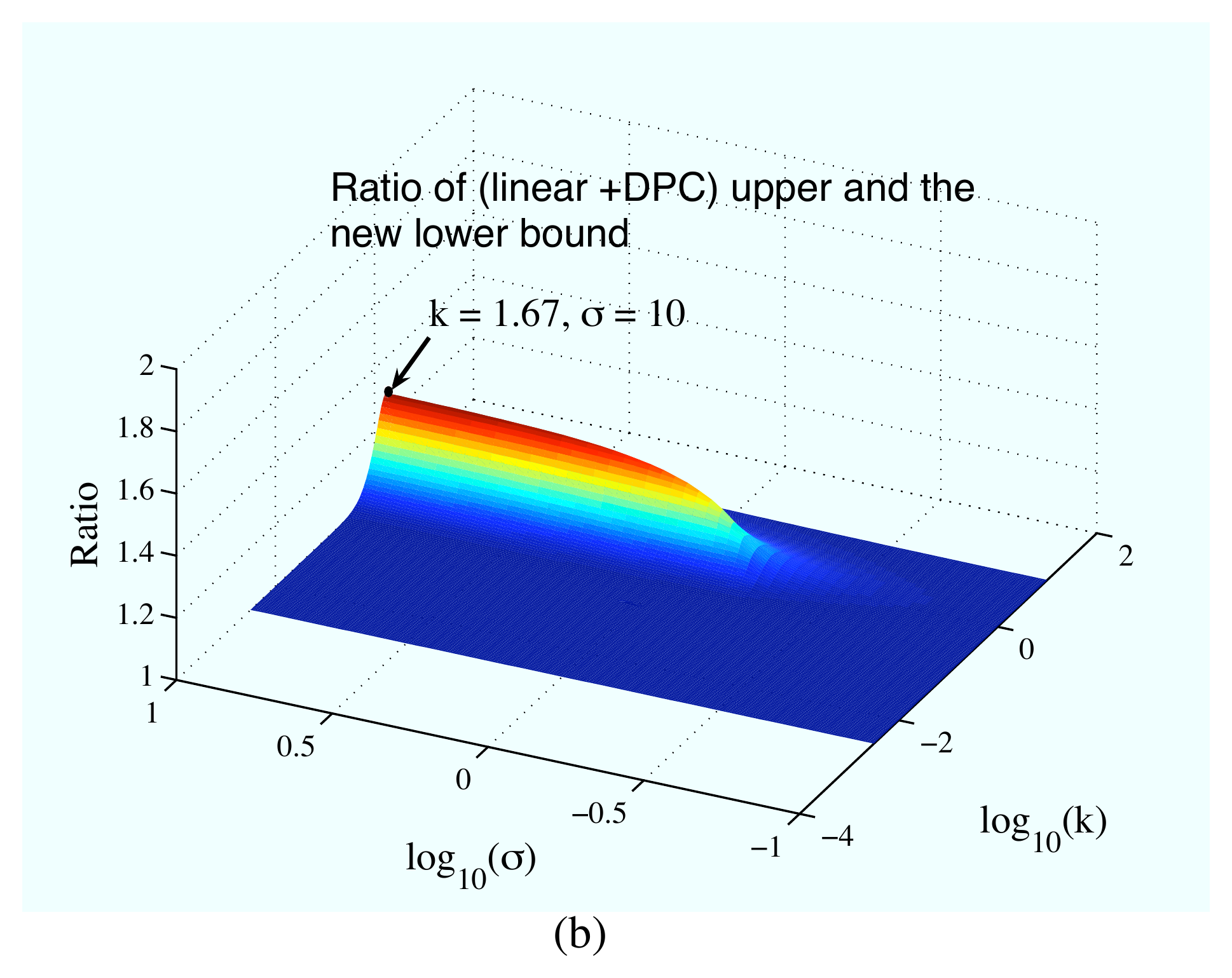}
\caption{The ratio of upper and lower bounds on the total asymptotic cost for the vector Witsenhausen counterexample with the lower bound taken from~\cite{WitsenhausenJournal} in (a) and from Corollary~\ref{coro:wit} in (b). As compared to the previous best known ratio of $2$~\cite{WitsenhausenJournal}, the ratio here is smaller than $1.3$. Further, an infinitely long ridge along $\sigma^2=\frac{\sqrt{5}-1}{2}$ and small $k$ that is present in lower bounds of~\cite{WitsenhausenJournal} is no longer present here. This is a consequence of the tightness lower bound at $MMSE=0$, and hence for small $k$. A ridge remains along $k\approx 1.67$ ($\log_{10}(k)\approx 0.22$) and large $\sigma$, and this can be understood by observing Fig.~\ref{fig:sigma} for $\sigma=10$.}
\label{fig:ratio}
\end{center}
\end{figure}

\begin{figure}[htb]
\begin{center}
\includegraphics[scale=0.55]{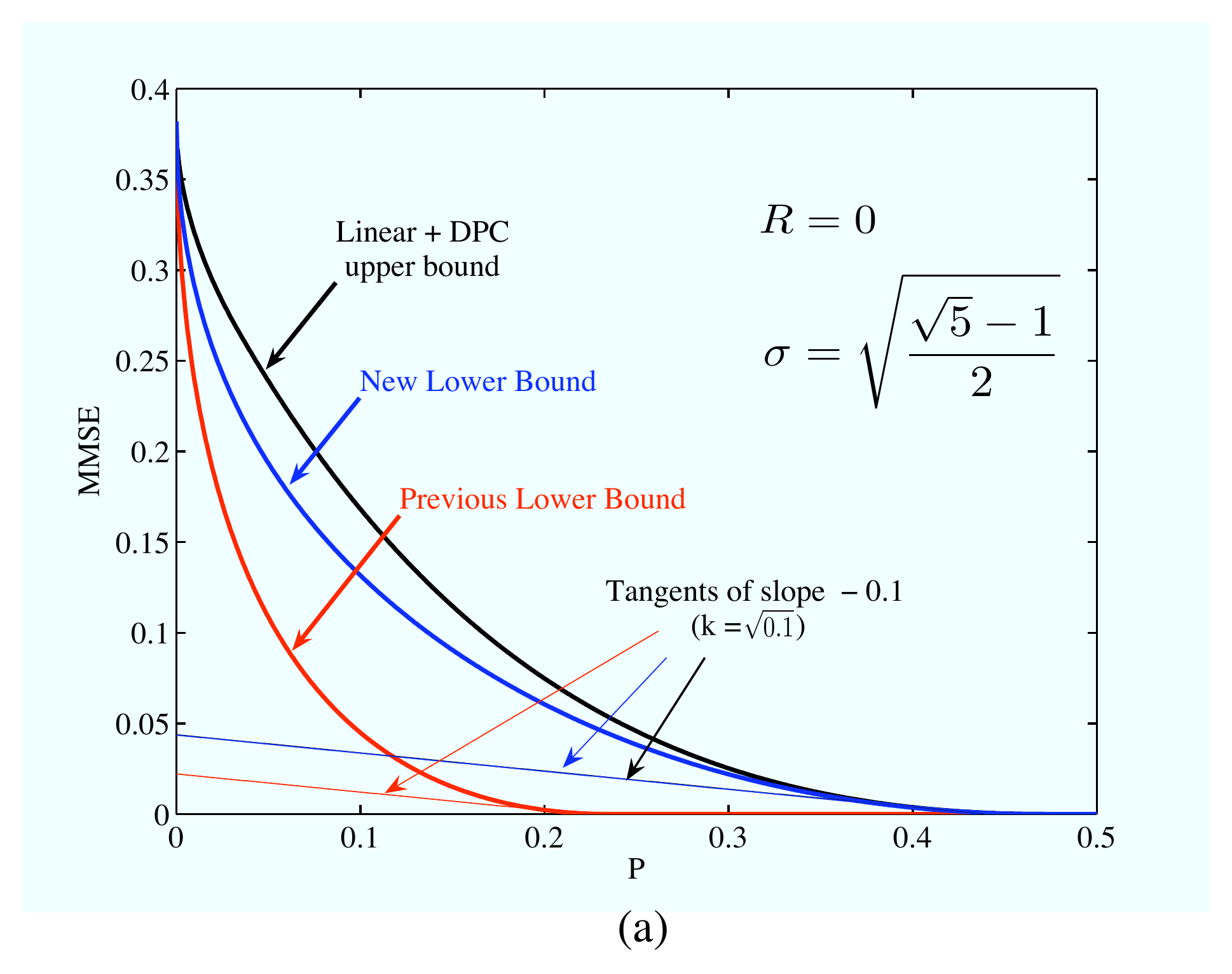}\\\includegraphics[scale=0.55]{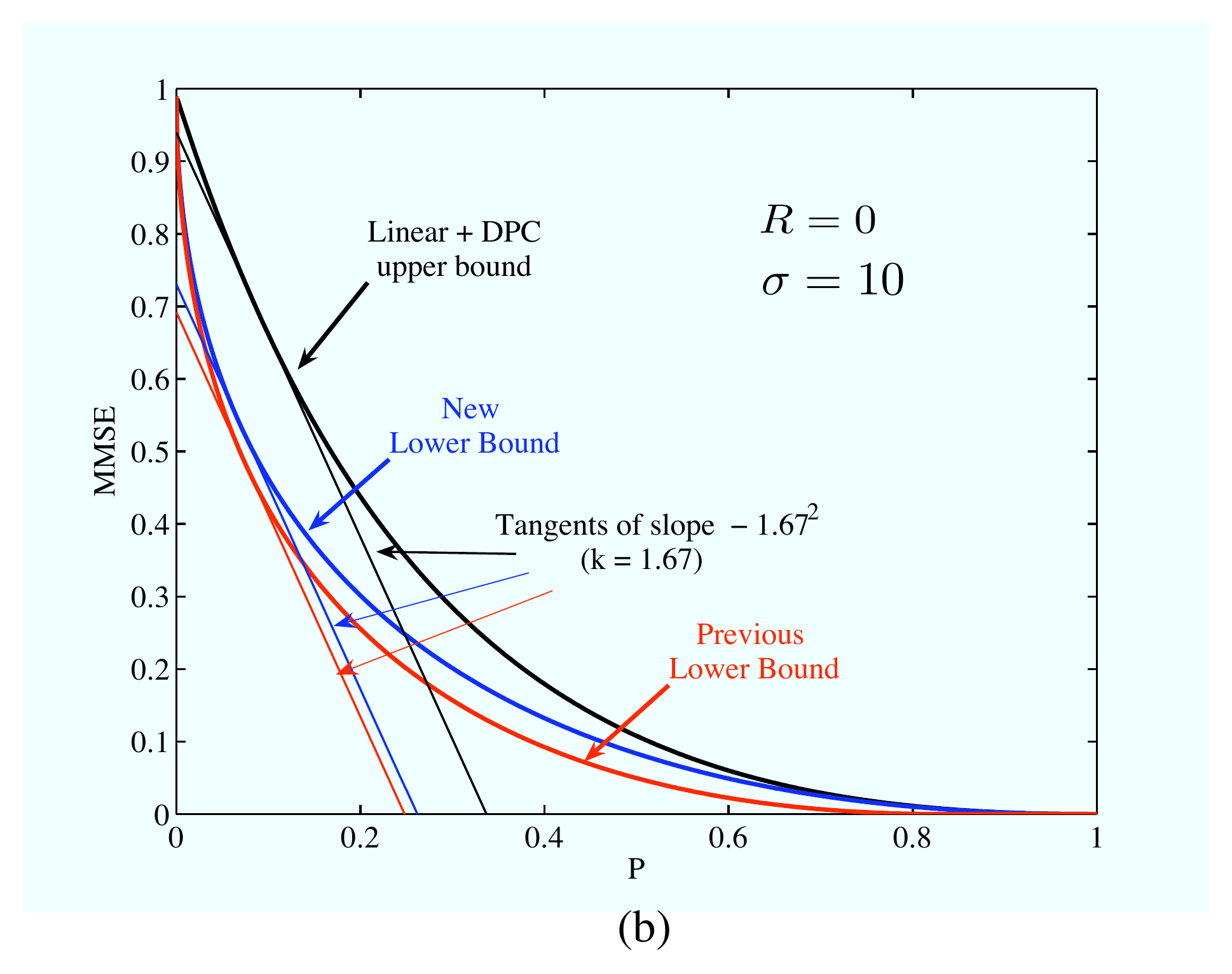}
\caption{Upper and lower bounds on  asymptotic $MMSE$ vs $P$ for $\sigma=\sqrt{\frac{\sqrt{5}-1}{2}}$ (square-root of the Golden ratio; Fig. (a)) and $\sigma=10$ (b) for zero-rate (the vector Witsenhausen counterexample). Tangents are drawn to evaluate the total cost for $k=\sqrt{0.1}$ for $\sigma=\sqrt{\frac{\sqrt{5}-1}{2}}$, and for  $k=1.67$ for $\sigma=10$ (slope $= -k^2$). The intercept on the $MMSE$ axis of the tangent provides the respective bound on the total cost. The tangents to the upper bound and the new lower bound almost coincide for small values of $k$. At $k\approx 1.67$ and $\sigma=10$, however, our bound is not  significantly better than that in~\cite{WitsenhausenJournal} and hence the ridge along $k\approx 1.67$ remains in the new ratio plot in Fig.~\ref{fig:ratio}.}
\label{fig:sigma}
\end{center}
\end{figure}


\begin{figure}[htb]
\begin{center}
\includegraphics[scale=0.47]{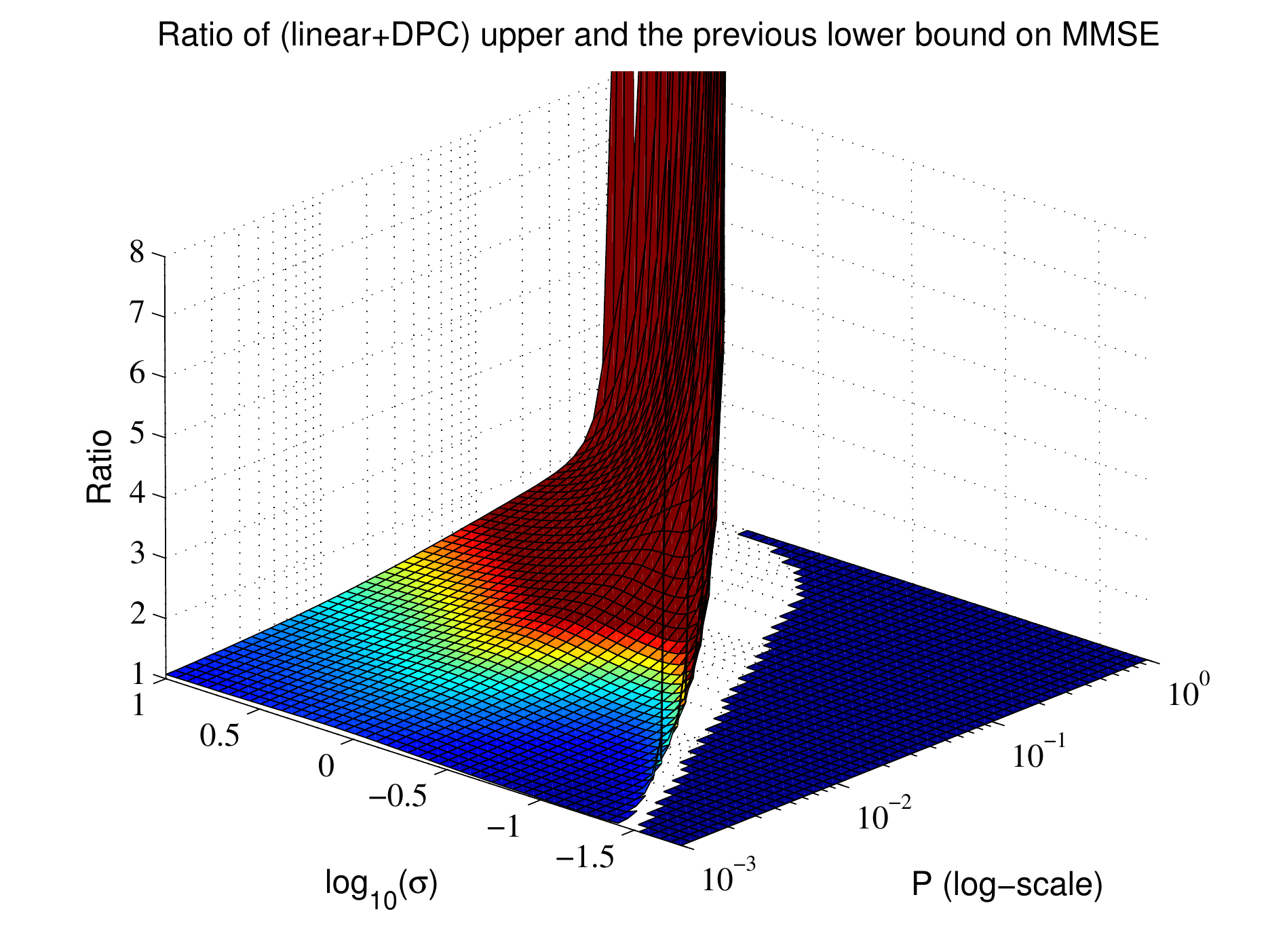}
\includegraphics[scale=0.47]{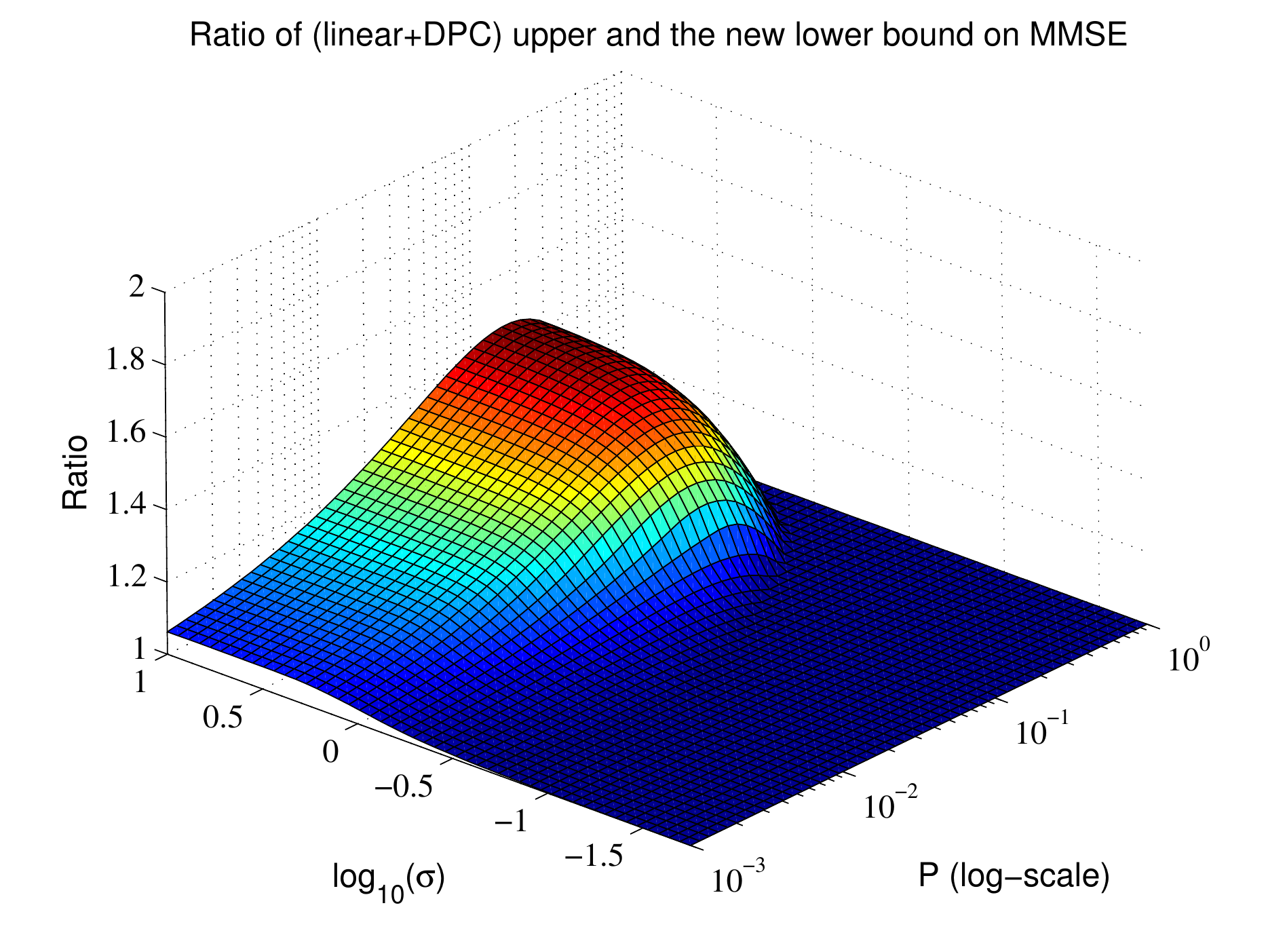}
\caption{Ratio of upper and lower bounds on $MMSE$ vs $P$ and $\sigma$ at $R=0$. Whereas the ratio diverges to infinity with the old lower bound of~\cite{WitsenhausenJournal} (top), it is bounded by $1.5$ for the new bound (bottom). This is a consequence of the improved tightness of the new bound at small $MMSE$.}
\label{fig:MMSEratio}
\end{center}
\end{figure}

\begin{figure}[htb]
\begin{center}
\includegraphics[scale=0.47]{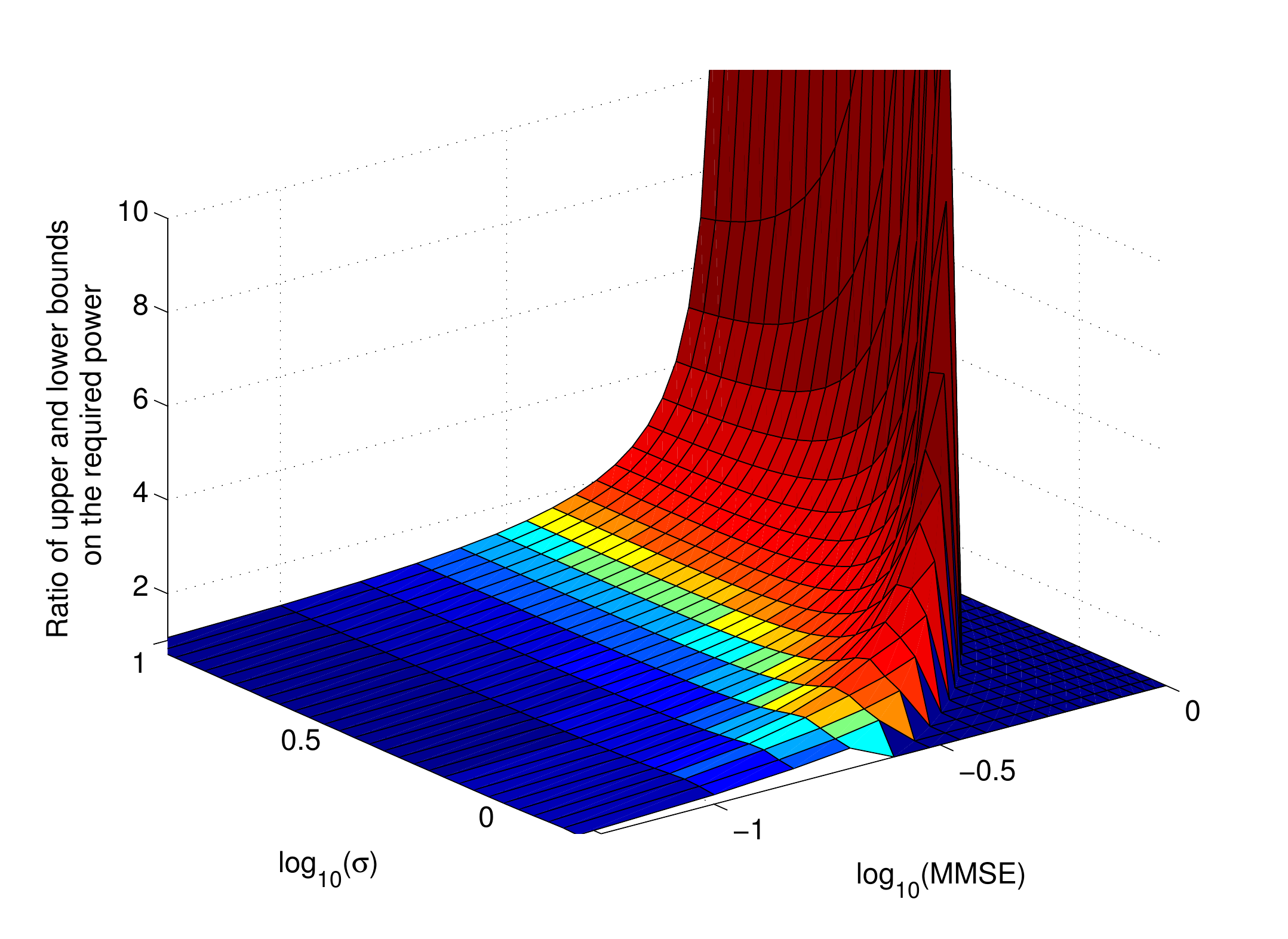}
\caption{Ratio of upper and lower bounds on $P$ vs $MMSE$ and $\sigma$ at $R=0$. Interestingly, the ratio increases to infinity as $\sigma\rightarrow\infty$ along the path where $P$ is close to zero (corresponding to ``high" $MMSE=\frac{\sigma^2}{\sigma^2+1}$).}
\label{fig:PowerRatio}
\end{center}
\end{figure}

\begin{figure}[htb]
\begin{center}
\includegraphics[scale=0.55]{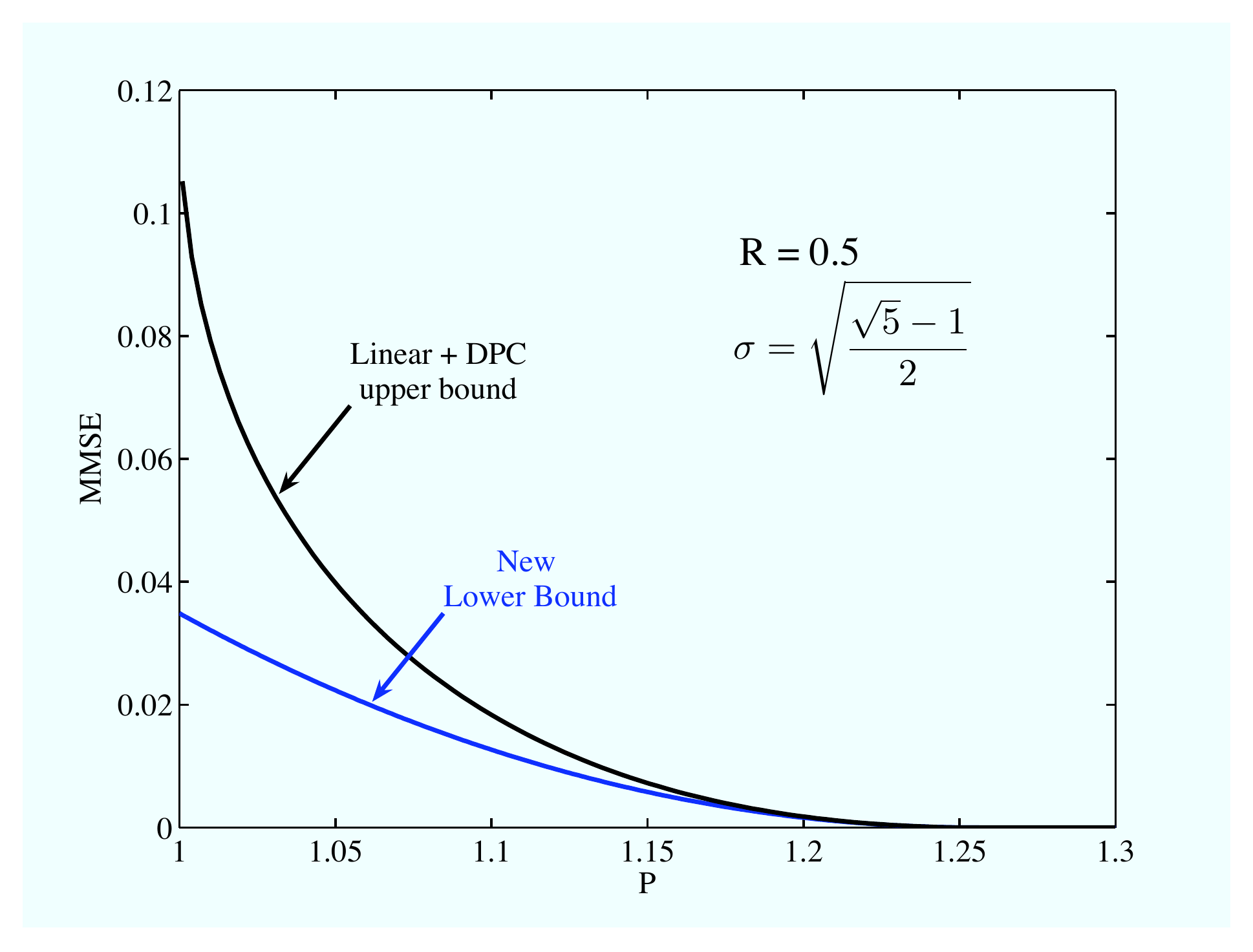}
\caption{Upper and lower bounds on $P$ vs $MMSE$ for $\sigma=\sqrt{\frac{\sqrt{5}-1}{2}}$\ for $R=0.5$. Though the bounds match at $MMSE=0$ (by Corollary~\ref{coro:match}), the bounds do not match at the minimum power ($P=1$ here) for nonzero rates. Below $P=1$, communication at $R=0.5$ is not possible.}
\label{fig:nonzerorate}
\end{center}
\end{figure}

\begin{figure}[htb]
\begin{center}
\includegraphics[scale=0.47]{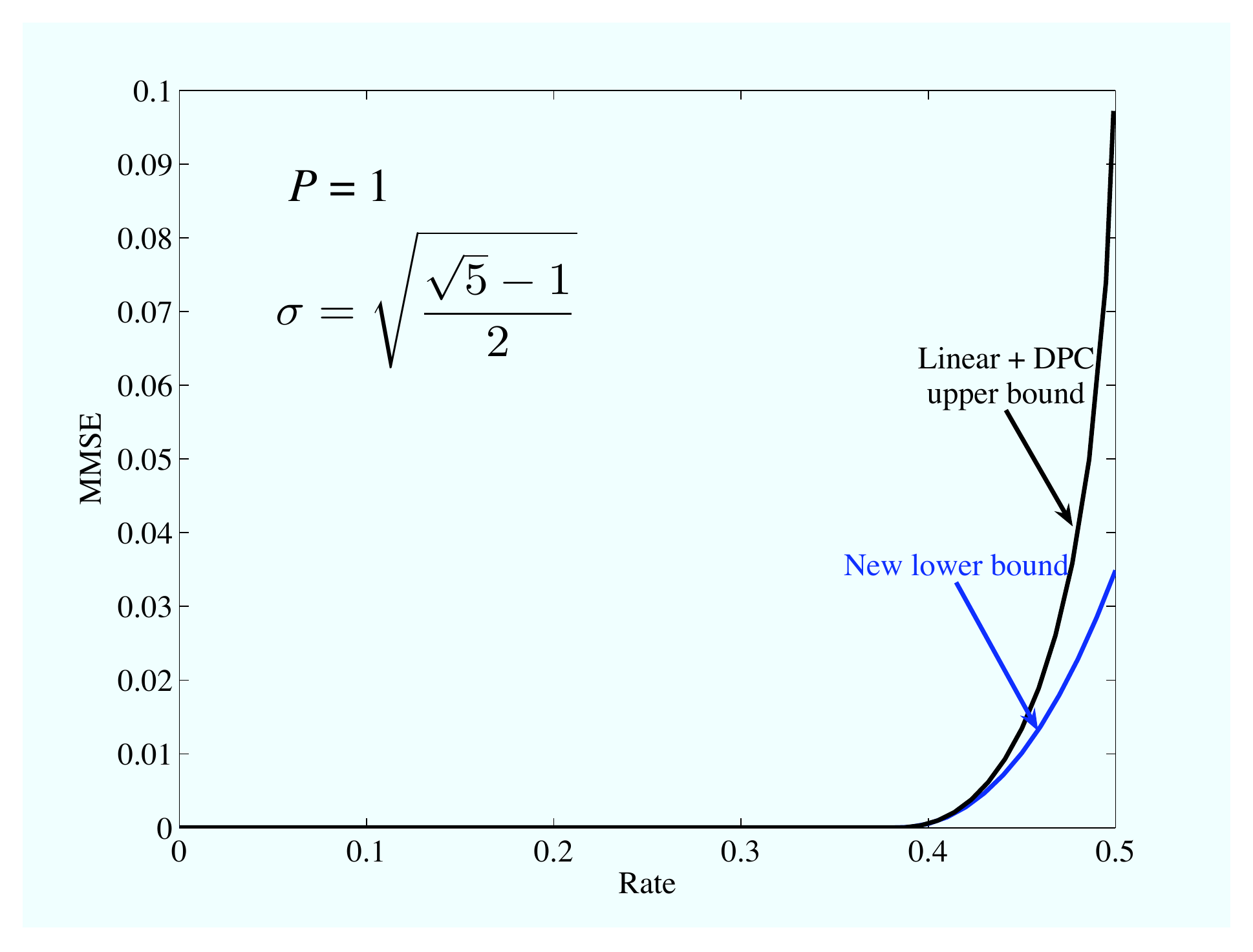}
\caption{Plot of upper and lower bounds on $MMSE$ vs rate for fixed power $P=1$ and $\sigma=\sqrt{\frac{\sqrt{5}-1}{2}}$. Higher rates require higher average distortion in the reconstruction of $\m{X}$.}
\label{fig:RvsMMSE}
\end{center}
\end{figure}

\vspace{-0.1in}

\section*{Acknowledgments}
P.~Grover and A.~Sahai acknowledge the support of the National Science Foundation (CNS-403427, CNS-093240, CCF-0917212 and CCF-729122) and Sumitomo Electric. A.~B.~Wagner acknowledges the support of NSF CSF-06-42925 (CAREER) grant. We thank Hari Palaiyanur, Se Yong Park and Gireeja Ranade for helpful discussions.

\bibliographystyle{IEEEtran}
\bibliography{IEEEabrv,MyMainBibliography,MyMainBibliography_2}

\end{document}